\documentclass[reqno,11pt,a4]{amsart}
\usepackage{amsmath,amssymb,tabularx,setspace,color}
\usepackage{pdflscape}
\newtheorem{theorem}{Theorem}

\makeatletter
\renewcommand\section{\@startsection {section}{1}{\z@}%
                                   {-3.5ex \@plus -1ex \@minus -.2ex}%
                                   {2.3ex \@plus.2ex}%
                                   {\normalfont\large\bfseries}}

\usepackage{multirow}
\usepackage[bf,small,tableposition=top]{caption}
\usepackage{amsthm}
\usepackage{lineno}
\usepackage[figuresright]{rotating}
\usepackage{enumerate}
\usepackage[T1]{fontenc}
\usepackage[utf8]{inputenc}
\usepackage{graphics}
\usepackage{graphicx}
\usepackage{multirow}
\usepackage{lscape}
\usepackage[dvipsnames]{xcolor}
\usepackage{mathrsfs}
\usepackage{hyperref}
\usepackage{textcomp}
\usepackage{listings}
\usepackage{amsmath}
\begin{document}
\doublespace
\title[]{Jackknife empirical likelihood based inference for Probability weighted moments}
\author[]{D\lowercase{eepesh} B\lowercase{hati}$^*$,   S\lowercase{udheesh} K K\lowercase{attumannil}$^{**,\dag}$ \lowercase{and} S\lowercase{reelakshmi}  N$^{***}$  \\
$^{*}$D\lowercase{epartment} \lowercase{of} S\lowercase{tatistics,} C\lowercase{entral} U\lowercase{niversity of } R\lowercase{ajasthan},  I\lowercase{ndia.}\\$^{**}$I\lowercase{ndian} S\lowercase{tatistical} I\lowercase{nstitute},
 C\lowercase{hennai}, I\lowercase{ndia.}\\
$^{***}$ I\lowercase{ndian} I\lowercase{nstitute of } T\lowercase{echnology},
C\lowercase{hennai}, I\lowercase{ndia.}}

\thanks{ {$^{\dag}$} {Corresponding E-mail: \tt skkattu@isichennai.res.in. } D\lowercase{eepesh} B\lowercase{hati} would like to thanks Indian Statistical Institute, Chennai for the support during his visit to
ISI Chennai.}
\begin{abstract}
In the present article,  we  discuss jackknife empirical likelihood (JEL) and adjusted jackknife empirical likelihood (AJEL) based inference for finding confidence intervals for probability weighted moment (PWM). We obtain the asymptotic distribution of the JEL ratio and AJEL ratio statistics. We compare the performance of the proposed confidence intervals with recently developed methods in terms of coverage probability and average length. We also develop JEL and AJEL based test for PWM and study it properties. Finally  we illustrate our method using rainfall data of Indian states.\\
 \noindent {\sc Keywords:} Empirical Likelihood; JEL; Probability weighted moment; $U$-statistics.
\end{abstract}
\maketitle

\section{Introduction}
\par Probability weighted moments (PWM) generalize the concept of moments of a probability distribution.It is generally used to estimate the parameters of extreme distributions of natural phenomena. Particularly, in the fields of Hydrology and Climatology, researchers use PWM for the estimation of parameters of the distributions related to water discharges and maxima of temperatures. Greenwood et al. (1979) proposed the concept of probability weighted moments  to  estimate the parameters involved in the models of extremes of natural phenomena.  The PWM of a random variable $X$  with distribution function  $F(.)$ is defined as
\begin{equation*}
\mathcal{M}_{p,r,s}=\mathbb{E}\left \lbrace X^pF^r(X)(1-F(X))^s\right\rbrace,
\end{equation*}
where $p$, $r$ and $s$ are any real numbers. Hosking et al. (1985) studied PWM of the form given by
\begin{equation}\label{pwm}
\beta_r=\mathbb{E}\left \lbrace XF^r(X)\right\rbrace,
\end{equation}
to characterize various distributional properties such as assessment of scale parameter, skewness of the distribution and L-moments. In this article we discuss the empirical likelihood inference of $\beta_r$.

 \par Given a random sample $X_1,X_2,\cdots, X_n$ of size $n$ from  $F$, let $X_{(i)},i=1,2,\ldots,n$ be the $i$-$th$ order statistic. David and Nagaraja (2003) proposed an estimator (D-N estimator) for $\beta_r$ by replacing $F$ in (\ref{pwm}) with its empirical version, $\hat{F}_n(x)=\frac{1}{n}\sum\limits_{i=1}^{n}\mathbb{I}_{\lbrace X_i\le x\rbrace}$, where $\mathbb{I}$ denotes the indicator function.  Their estimator is given by
\begin{equation}\label{1}
\overline{\beta}_r=\frac{1}{n}\sum\limits_{i=1}^{n}\left(\frac{i}{n}\right)^r X_{(i)}.
\end{equation}
To develop an empirical likelihood inference for $\beta_r$,  Vexler et al. (2017) proposed an estimator (Vexler's estimator) given by
\begin{equation}\label{2}
\tilde{\beta}_r= \frac{1}{r+1}\sum\limits_{i=1}^{n} X_{(i)} \left\lbrace  \left(\frac{i}{n} \right)^{r+1} -\left(\frac{i-1}{n} \right)^{r+1}  \right\rbrace
\end{equation}
and showed that asymptotic behaviour of both $\tilde{\beta}_r$ and $\overline{\beta}_r$ are same.


 Empirical likelihood is a non-parametric inference tool which make use of likelihood principle. This inference procedure is  firstly used by Thomas and Grunkemeier (1975) to obtain the confidence interval for survival probability when data contain censored observations. Pioneering papers by Owen (1988, 1990) for finding the confidence interval of regression parameters take the empirical likelihood method into a general methodology and have wide applications in many statistical areas.  This approach enjoys the wide acceptance among the researchers as it combines the effectiveness of the likelihood approach with the reliability of non-parametric procedure.

 Empirical  likelihood finds applications in regression, survival analysis and inference for income inequality measures [Shi and Lau (2000),  Whang (2006), Qin et al. (2010), Peng (2011), Qin et al. (2013), Zhou (2015), Wang et al.(2016), Wang and Zhao (2016)]. Recently, Vexler et al. (2017) proposed an empirical likelihood based inference for PWM and showed  that the limiting distribution of log empirical likelihood ratio statistic is $\chi^{2}$ distribution with one degree of freedom and obtained confidence intervals and likelihood ratio tests for PWM.

\par   In empirical likelihood approach,  we need to maximize the non-parametric likelihood function subject to some constraints. When the constraints are linear, the maximization of the likelihood is not difficult. However, when the constraints are based on nonlinear statistics such as $U$-statistics with higher degree $(\ge 2)$ kernel the implementation of empirical likelihood becomes challenging. To overcome this difficulty, Jing et al. (2009) introduced the jackknife empirical likelihood (JEL) inference, which combines two of the popular non-parametric approaches namely, the jackknife and the empirical likelihood approach. Chen et al. (2008) proposed the concept of adjusted empirical likelihood which preserves the asymptotic properties of empirical likelihood. Even though it is an adjustment given to empirical likelihood, adjusted empirical likelihood procedure yields better coverage probability than bootstrap calibration and Bartlett correction methods.  Recently Zhao et al. (2015) introduced adjusted jackknife empirical likelihood (AJEL) inference so that restriction on parameter values on the convex hull of estimating equation is relaxed.


\par Motivated by these recent works, in this article, we develop JEL and AJEL based inference to construct confidence intervals and likelihood ratio tests for PWM. The present article is structured as follows. In Section 2, we derive the jackknife empirical log likelihood ratio for $\beta_r$ and obtained its limiting distribution. Using this result, we construct jackknife empirical likelihood   based confidence interval and likelihood ratio test for $\beta_r$.  We further derive AJEL based confidence interval and likelihood ratio test for $\beta_r$. In Section 3, a comparison of the proposed methods with empirical likelihood method are given using Monte Carlo simulation. Finally an illustration of our methods using a real data is given in Section 4. Major findings of the study are given in Section 5.

\section{Jackknife empirical likelihood inference for PWM}
In this section, first we discuss JEL based inference for confidence interval and likelihood ratio test for $\beta_r$. Later we discuss same problems using AJEL based inference. The implementation of these methods require jackknife pseudo values obtained using an estimator of $\beta_r$. For this purpose, we introduce an estimator of $\beta_r$ using theory of $U$-statistics.

To obtain a $U$-statistic based estimator for $\beta_r$ we rewrite (\ref{pwm}) as
\begin{align*}
\beta_{r}=& \frac{1}{r+1}\int\limits_{-\infty}^{\infty} (r+1)xF^r(x) dF(x) \\
=& \frac{1}{r+1}\cdotp \mathbb{E}(\max(X_1,X_2,\cdots,X_{r+1})),
\end{align*}provided $r$ is a positive integer.
Therefore an unbiased  estimator of $\beta_r$  is given by
\begin{equation}\label{est}
\widehat{\beta}_{r}=\frac{1}{r+1} \frac{1}{C_{m,n}} \sum\limits_{C_{m,n}}h{(X_{i_1},X_{i_2},\cdots,X_{i_{r+1}})},
\end{equation}
where $h{(X_{i_1},X_{i_2},\cdots,X_{i_{r+1}})}=\max(X_{i_1},X_{i_2},\cdots,X_{i_{r+1}})$ and the summations is over the set  $C_{m,n}$ of all combinations of $(r+1)$ distinct elements $\lbrace i_1,i_2,\cdots,i_{r+1}\rbrace$ chosen from $\lbrace 1,2,\cdots,n\rbrace$. Clearly $\widehat{\beta}_{r}$ is a consistent estimator of ${\beta}_{r}$ (Lehmann, 1951). Use of consistent and unbiased estimator to construct an empirical likelihood based confidence interval give better coverage  probability and average length. However, $\beta_r$  has a $U$-statistics based estimator with kernel of degree greater than one which results in non-linear constraints in the optimization problem associated with empirical likelihood. This makes the implementation of empirical likelihood theory very difficult. This leads us to construct JEL based confidence interval and test for $\beta_r$.

\par Next we discuss how to derive the jackknife empirical likelihood ratio for $\beta_r$.  The jackknife pseudo-values for $\beta_r$ is given by
\begin{equation}\label{jsv}
\widehat{V}_{k}= n \widehat{\beta}_{r}-(n-1)\widehat{\beta}_{r,k}; \qquad k=1,2,\cdots,n,
\end{equation}
where $\widehat{\beta}_{r,k}$ is the estimator of $\beta_r$ obtained from (\ref{est}) by using $(n-1)$ observations $X_1,X_2,...,X_{k-1},X_{k+1},...,X_n$. The jackknife estimator $\widehat{\beta}_{r,jack}$ of $\beta_r$ is the average of the jackknife pseudo-values, that is
$$\widehat{\beta}_{r,jack}=\frac{1}{n}\sum\limits_{k=1}^{n}\widehat{V}_{k}.$$
As the pseudo-values are constructed using a $U$-statistic, the two  estimators $\widehat{\beta}_{r}$ and $\widehat{\beta}_{r,jack}$ coincide. We use the jackknife pseudo-values $\widehat{V}_{k}$ defined in equation (\ref{jsv}) to construct JEL  for $\beta_r$.
The jackknife empirical likelihood of $\beta_{r}$ is defined as
\begin{equation}\label{6}
 J(\beta_{r})=\sup_{\bf p} \left(\prod_{k=1}^{n}{p_k};\,\, p_k \ge 0;\, \, \sum_{k=1}^{n}{p_k}=1;\,\,\sum_{k=1}^{n}{p_k (\widehat{V}_k}-\beta_r)=0\right),
\end{equation}
 where ${\bf p}=(p_1,p_2,...,p_n)$ is a probability vector. The maximum of (\ref{6}) occurs at
\begin{equation*}
  p_k=\frac{1}{n}\left(1+\lambda(\widehat{V}_{k}-\beta_r)\right)^{-1}, k=1,2,...,n,
\end{equation*}
where $\lambda$ is the solution of
\begin{equation}\label{eq7}
  \frac{1}{n}\sum_{k=1}^{n}{\frac{\widehat{V}_{k}-\beta_r}{1+\lambda (\widehat{V}_{k}-\beta_r)}}=0,
\end{equation}provided
\begin{equation*}
  \min_{{1\le k\le n}}\widehat{V}_{k}<\widehat \beta_r<  \max_{1\le k\le n}\widehat{V}_{k}.
\end{equation*}
Also note that, $\prod\limits_{k=1}^{n}p_i$, subject to $\sum\limits_{i=1}^{n}p_i=1$, attains its maximum $n^{-n}$ at $p_i=n^{-1}$. Hence, the jackknife empirical log-likelihood ratio  for  $\beta_r$ is given by
\begin{equation}\label{jelrat}
  l(\beta_r)=-\sum_{i=1}^{n}\log\left[1+\lambda (\hat{V}_{k}-\beta_r)\right].
\end{equation}

Next theorem explains the limiting distribution of $l(\beta_r)$ which can be used to construct the JEL based confidence interval and test for $\beta_r$.

\begin{theorem}
  Suppose that $E\left(h^2(X_{1}, X_{2},..., X_{r+1})\right)<\infty$ and $\sigma^2=Var(g(X))>0$, where $g(x)=E\left(h(X_1, X_{2},..., X_{r+1})|X_1=x\right)-\beta_{r}$. Then, as $n\rightarrow\infty$, the distribution of $-2l(\beta_r)$ is $\chi^{2}(1)$.
\end{theorem}
\begin{proof}
Let $S=\frac{1}{n}\sum_{k=1}^{n}(\widehat{V}_{k}-\beta_r)^{2}$.
Since $\widehat\beta_r=\frac{1}{n}\sum_{k=1}^{n}\widehat{V}_k$, by strong law of large numbers we obtain
\begin{equation}\label{eq8}
      S= \sigma^{2}+o(1).
      \end{equation}
Using Lemma A.4 of Jing et al. (2009) we have
      \begin{equation}\label{eq9}
    \max_{1\le k\le n} |\widehat{V}_{k}-\beta_{r}|=o(\sqrt{n}).
  \end{equation}
Hence using (\ref{eq8}) and (\ref{eq9}) we have
\begin{equation}\label{eq111}
  \frac{1}{n}\sum_{k=1}^{n}|\widehat{V}_{k}-\beta_{r}|^{3}\le  |\widehat{V}_{k}-\beta_{r}|\frac{1}{n}\sum_{k=1}^{n}(\widehat{V}_{k}-\beta_{r})^{2}=o(\sqrt{n}).
\end{equation}
The $\lambda$ satisfying  the equation (\ref{eq7}) has the property  (Jing et al. 2009)
\begin{equation}\label{lo}
  |\lambda|=O_{p}(n^{-\frac{1}{2}}).
\end{equation}
Hence using (\ref{eq9}) we obtain
\begin{equation}\label{eq10}
    \max_{1\le k\le n} \lambda|\widehat{V}_{k}-\beta_{r}|=o({1}).
  \end{equation}Therefore, using equations (\ref{eq111}), (\ref{eq10}) and (\ref{lo}) we have
  \begin{equation*}
  \frac{1}{n}\sum_{k=1}^{n}(\widehat{V}_{k}-\beta_{r})^{3}\lambda^2|1+\lambda(\widehat V_k-\beta_r)|^{-1}=o_p(\sqrt{n})O_p({1/n})O_p(1)=o_p(1/\sqrt{n}).
\end{equation*}Hence from (\ref{eq7}) we obtain
\begin{equation}\label{eq11}
\lambda=\frac{(\widehat{\beta}_{r}-\beta_{r})}{S}+o_p(1/\sqrt n).
\end{equation}
Using Taylor's theorem, we can express  $l(\beta_r)$ given in (\ref{jelrat}) as
\begin{equation}\label{asylike}
  -2l(\beta_r)=2n\lambda(\widehat{\beta}_{r}-\beta_{r})-nS\lambda^2+R(\beta_r),
\end{equation}where $R(\beta_r)$ is the reminder term.  Using $|\lambda|=O_{p}(n^{-\frac{1}{2}})$ and (\ref{eq111}) we obtained  the reminder term $R(\beta_r)=o_p(1)$.  Hence using (\ref{eq11}), we can express (\ref{asylike}) as
\begin{equation}\label{eq15}
  -2l(\beta_r)=\frac{n(\widehat{\beta}_{r}-\beta_{r})^2}{S}+o_p(1).
\end{equation}
Using the central limit theorem for $U$-statistics, as $n\rightarrow \infty$,   the asymptotic distribution  of $\sqrt{n}(r+1)\left(\widehat{\beta}_{r}-{\beta}_{r}\right)$ is Gaussian with mean zero and variance $(r+1)^2\sigma^2$, where
 \begin{eqnarray*}
 \sigma^{2}&=&Var \left(E(max(X_1,X_2,\cdots,X_{r+1})|X_1=x)\right)\\&=&Var\left(XF^{r}(X)+r\int_{X}^{\infty}yF^{r-1}(y)dF(y)\right).
 \end{eqnarray*}
   Hence, as $n \to \infty$, $\sqrt{n}\left(\widehat{\beta}_{r}-{\beta}_{r}\right)$ converges in distribution to normal with mean zero and variance $\sigma^2$.  Accordingly $\frac{n(\widehat{\beta}_{r}-{\beta}_{r})^2}{\sigma^{2}}$  converges in distribution to $ \chi^{2}$ with one degree of freedom. In view of (\ref{eq8}), by Slutsky's theorem, from (\ref{eq15}) we have the result.\\
\end{proof}

 Using Theorem 1, JEL based confidence interval for $\beta_r$ at $100(1-\alpha)\%$ is given by
\begin{equation*}
  CI_1=\left\{\beta_r|-2l(\beta_r)\le \chi^2_{1,1-\alpha}\right\},
\end{equation*}where $\chi^2_{1,1-\alpha}$ is the $(1-\alpha)$th percentile of chi-square distribution with one degree of freedom. The performance of these confidence intervals in terms of coverage probabilities and average lengths were evaluated via a Monte Carlo simulation and the results are reported in Section 3.

Using the asymptotic distribution of jackknife empirical  log likelihood ratio we can develop JEL based test for testing the hypothesis $\beta_r=\beta_r^0$, where $\beta_r^0$ is a specific value of $\beta_r$.  We reject the null hypothesis at significance level $\alpha$ if $$-2l(\beta_r)> \chi^2_{1,1-\alpha}.$$ Simulation study shows that the type 1 error rate of the test converges to desired significance level and has good power. The results of the simulation study are also reported in Section 3.

Next we discuss construction of AJEL based confidence interval for $\beta_r$.    Define
\begin{equation}
\widehat{V}_{n+1}=-\frac{a_n}{n}\sum\limits_{k=1}^{n} \widehat{V}_{k},
\end{equation}
for some positive $a_n$. Chen et al. (2008) suggested to take $a_n=\max\lbrace 1,\log_e (n/2)\rbrace$.  The  adjusted jackknife estimator of $\beta_r$ is defined as
\begin{equation}
\widehat{\beta}_{r,adjjack}=\frac{1}{n+1}\sum\limits_{k=1}^{n+1}\widehat{V}_{k}.
\end{equation}
The adjusted jackknife empirical  likelihood of $\beta_r$ is given by
\begin{equation}\label{jel}
 R(\beta_{r})=\sup_{\bf p} \left(\prod_{k=1}^{n+1}{(n+1) p_k};\, \, p_k \ge 0; \,\, \sum_{k=1}^{n+1}{p_k}=1;\,\,\sum_{k=1}^{n+1}{p_k (\widehat{V}_k}-\beta_r)=0\right).
\end{equation}
The maximum of (\ref{6}) occurs at
\begin{equation*}
  p_k=\frac{1}{n+1}\left(1+\lambda_1(\widehat{V}_{k}-\beta_r)\right)^{-1}, k=1,2,...,n+1,
\end{equation*}
where $\lambda_1$ is the solution of
\begin{equation*}
  \frac{1}{n+1}\sum_{k=1}^{n+1}{\frac{\widehat{V}_{k}-\beta_r}{1+\lambda_1 (\widehat{V}_{k}-\beta_r)}}=0.
\end{equation*}
Hence, the adjusted  jackknife empirical log likelihood ratio  for  $\beta_r$ is given by
\begin{equation*}
  l_1(\beta_r)=-\sum_{k=1}^{n+1}\log\left[1+\lambda_1 (\hat{V}_{k}-\beta_r)\right].
\end{equation*}

\begin{theorem}
Under the assumption of  Theorem 1, and for $a_n = o_p(n^{2/3})$,  as $n\rightarrow \infty$,  $-2l_{1}(\beta_{r})$, is distributed as $\chi^2(1)$.
\end{theorem}
\begin{proof}
The proof follows on similar lines of proof of Theorem 1.
Note that as long as $a_n = o_p(n)$, we have $|\lambda_1| = O_p(1/\sqrt n)$.
Consider
\begin{eqnarray*}\label{ajel}
  -2l_1(\beta_r)&=&2\sum_{k=1}^{n+1} \log(1+\lambda \widehat V_k)
  \\&=&2\sum_{k=1}^{n+1}\left(\lambda\widehat V_k- \lambda^2\widehat V_{k}^{2}/2\right)+ o_p(1)\\&=& 2n\lambda(\widehat{\beta}_{r}-\beta_{r})-nS\lambda^2+o_p(1)\\&=&
  \frac{n(\widehat{\beta}_{r}-\beta_{r})^2}{S}+o_p(1),
\end{eqnarray*} where the second last identity follows from the fact that the $(n + 1)$th term of the summation is $a_n O_p(n^{-3/2})=o_p(n)O_p(n^{-3/2})=o_p(1)$. Hence by Slutsky's theorem we have the result.
\end{proof}
 Using the asymptotic distribution of adjusted jackknife empirical log likelihood ratio we can construct a confidence interval and likelihood ratio test for $\beta_r$. A $100(1-\alpha) \%$ AJEL based confidence interval for $\beta_r$ is given by
\begin{equation*}
  CI_2=\left\{\beta_r|l_{1}(\beta_{r})\le \chi^2_{1,1-\alpha}\right\}.
\end{equation*}
In AJEL based ratio test for testing the hypothesis $\beta_r=\beta_r^0$, we reject the hypothesis when $ -2l_{1}(\beta_{r})> \chi^2_{1,1-\alpha}$.  The performance of these confidence intervals and tests are also evaluated in Section 3.

\begin{figure}[htp]
\centering
\includegraphics[width=150mm]{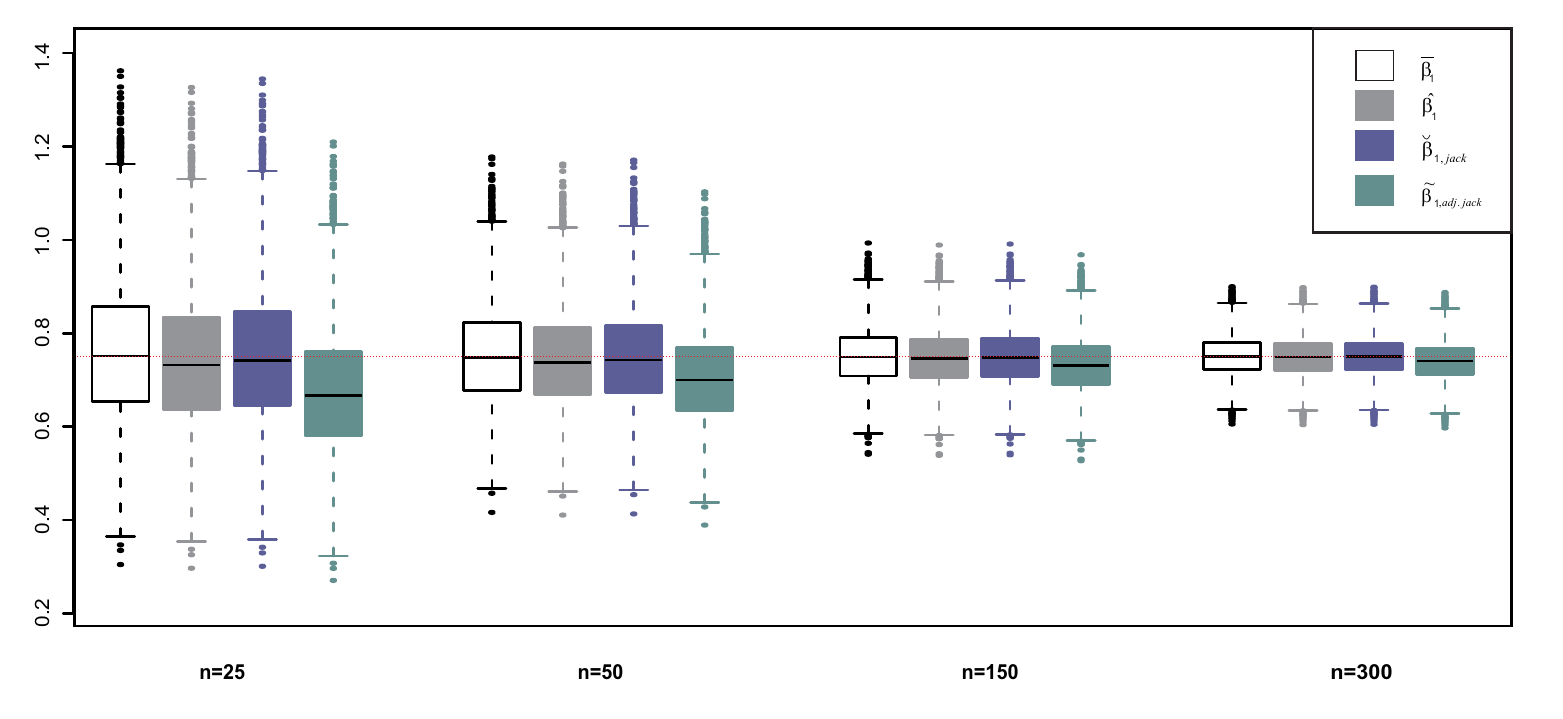}
\caption{Box plot of estimated $\beta_1$ for Exp(1) obtained by different estimator for different sample sizes.}
\end{figure}

\begin{figure}[htp]
\centering
\includegraphics[width=150mm]{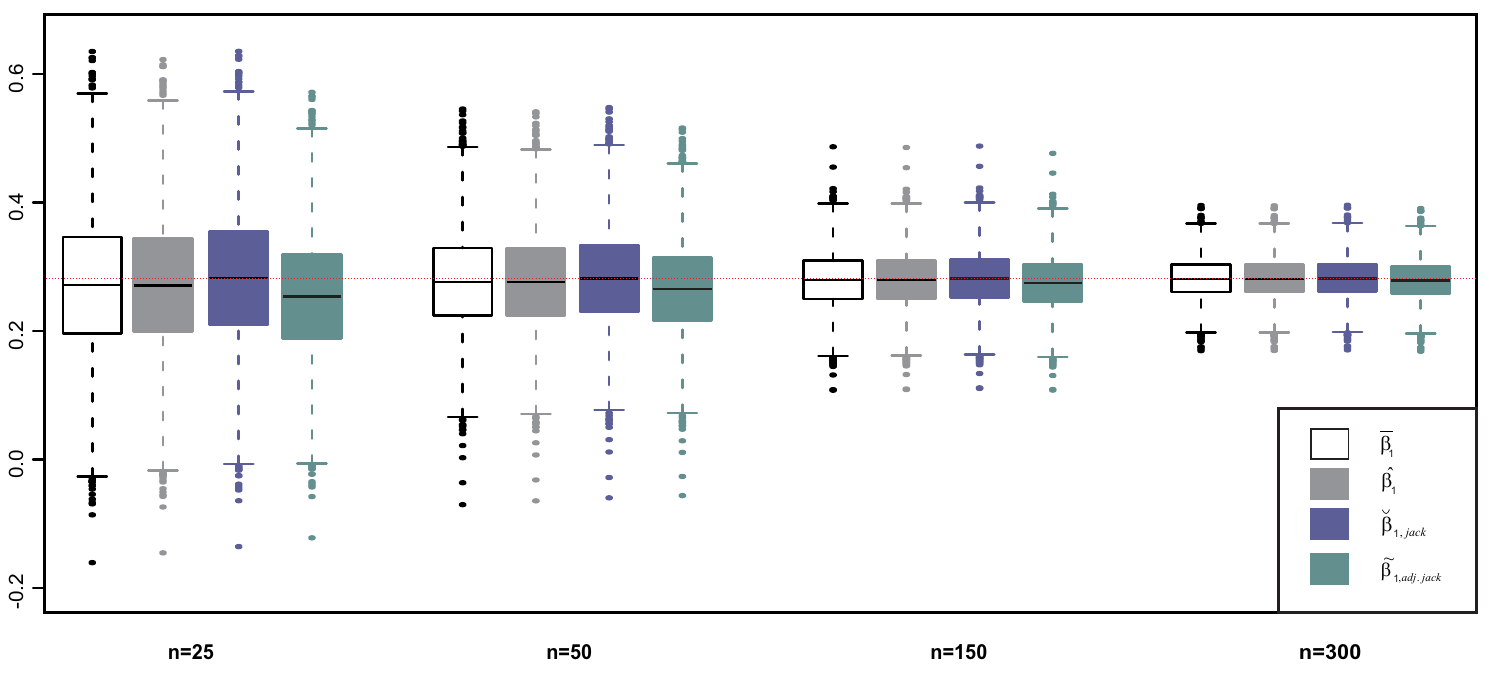}
\caption{Box plot of estimated $\beta_1$ for Normal(0,1) obtained by different estimator for different sample sizes.}
\end{figure}

\begin{figure}[htp]
\centering
\includegraphics[width=150mm]{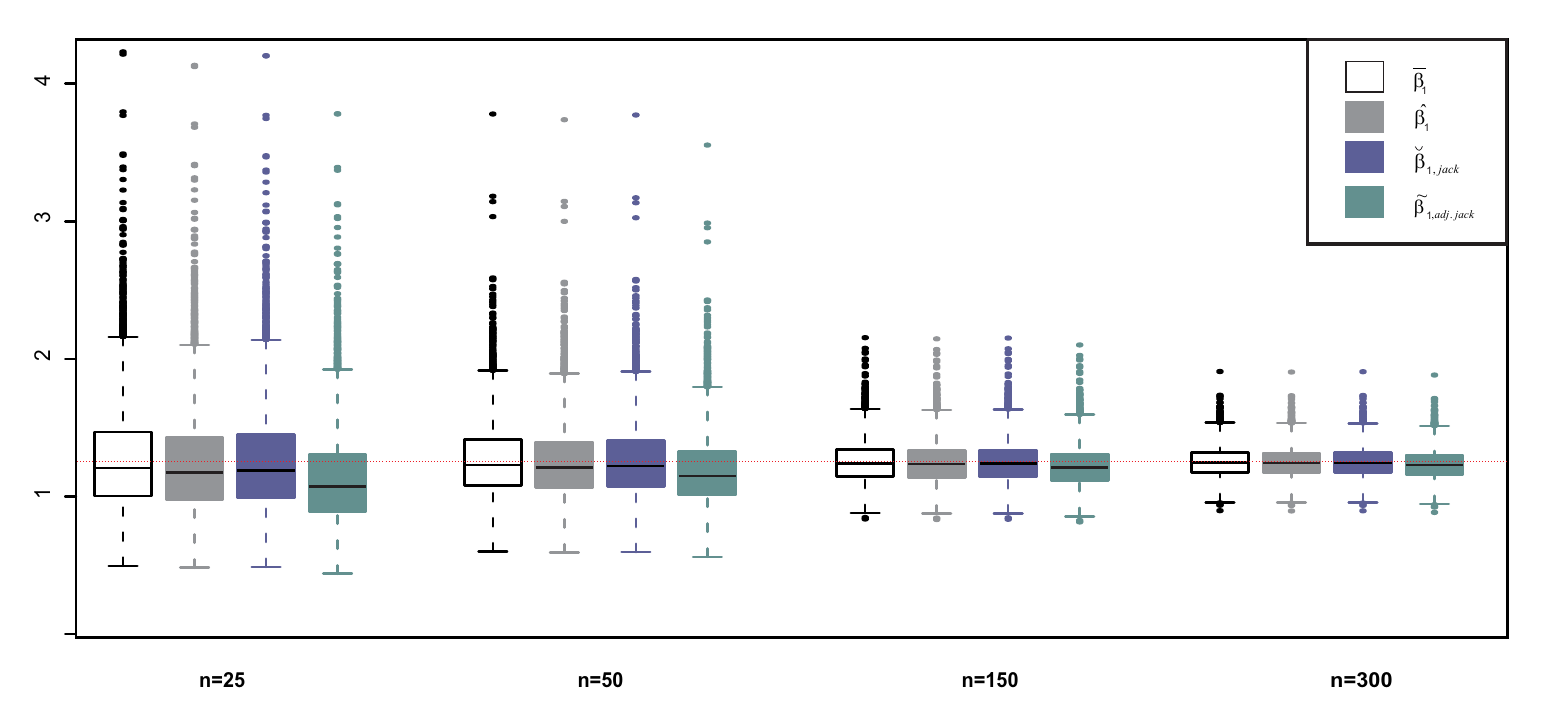}
\caption{Box plot of estimated $\beta_1$ for Lognormal(0,1) obtained by different estimator for different sample sizes.}
\end{figure}

The box plots of  four estimators $\overline{\beta}_1$, $\tilde{\beta}_1$, $\widehat{\beta}_{1,jack}$ and $\widehat{\beta}_{1,adjjack}$ obtained for different sample sizes $n=25,50, 150, 300$ are shown in Figures 1-3.   In Figure 1, we generate observations from standard exponential distribution and computed the four estimators for $n=25,50, 150, 300$.  The Monte-Carlo simulation procedure is repeated 5000 times. The red horizontal line across the box plot represents the actual value of $\beta_1$ when $X$  follows standard exponential. Similarly in Figure 2 and 3 show box plots of the estimators when the observations are generated from standard normal and standard lognormal distributions respectively.

\section{Monte carlo study}
\par In this section, we perform a Monte Carlo simulation study to compare the performance of  JEL and AJEL based confidence intervals with the confidence intervals obtained by empirical likelihood method constructed using D-N estimator (DNEL) as well as Vexler's estimator (VXL). The simulation is done using R.  We find empirical type 1 error as well as power of the proposed JEL and AJEL based tests. In our study, the observations are simulated from standard exponential, standard normal and standard log normal distributions with different sample sizes, $n=25, 50,100, 200$ and $300$.
The simulation procedure is repeated for five thousand times.  Note that we need to find the values of $\lambda$ and $\lambda_1$ to construct the confidence intervals and to obtain the critical regions of the proposed tests. We used the R  functions $uniroot$ and $optimize$ to find the values of $\lambda$ and $\lambda_1$.

The Table 1 gives the Monte Carlo variances of the four estimators  $\overline{\beta}_{r}$, $\tilde{\beta}_{r}$, $\widehat{\beta}_{r,jack}$ and $\widehat{\beta}_{r,adjjack}$ calculated based on the observations generated from standard exponential, standard normal and standard log normal distributions, respectively. We can see that $\widehat{\beta}_{r,adjjack}$ provides the least variance in almost all cases.

In Table 2 we report the coverage probability obtained for the four confidence intervals mentioned above.  From Table 2 it is clear that in all three cases, the coverage probabilities  of all four confidence intervals converge to the actual target value $(0.95)$. Hence it is worth to compare the average length of these intervals. In Table 3 we report the average length of the confidence intervals obtained for the confidence level $0.95$.  When $X$ has standard exponential and standard normal distribution  JEL based confidence intervals performs better than the other confidence intervals in terms of average length. In log normal case, in most of the cases JEL and AJEL based   confidence intervals  have smaller length compared  to other two confidence intervals.

\par The empirical type 1 error of the proposed tests are listed in Table 4. We find the type 1 error rate  for $\beta_r$ for $r=1,2,3,4$ when the samples are generated from above listed three distributions. For small sample sizes $n=25,50$,  AJEL based test has well controlled  type 1 error rates  for normally distributed data.  From Table 4, it is clear that for all the cases, empirical type 1 error reaches the nominal value $\alpha=0.05$ as the sample size increases.
\onehalfspace
Tables 5, 6 and 7 provide the power comparison for the JEL, AJEL, DNEL and VXL based tests when the null hypothesis is related to the following cases.
\begin{enumerate}
\item $\beta_1=0.68 $, $\beta_2= 0.55$, $\beta_3=0.47 $, $\beta_4= 0.41$ (or $X \sim Exp(0.9)$)
\item $\beta_1= 0.6$, $\beta_2=0.49 $, $\beta_3= 0.42$, $\beta_4=0.37 $ (or $X \sim Exp(0.8)$)
\item $\beta_1=0.56 $, $\beta_2=0.56 $, $\beta_3= 0.51$, $\beta_4=0.47 $ (or $X \sim Normal(0,4)$)
\item $\beta_1=0.84 $, $\beta_2=0.84 $, $\beta_3=0.77 $, $\beta_4=0.70 $ (or $X \sim Normal(0,9)$)
\item $\beta_1=2.63 $, $\beta_2=2.36 $, $\beta_3=2.16 $, $\beta_4=2.01 $ (or $X \sim  Lognormal(0,1.5)$)
\item $\beta_1=6.81 $, $\beta_2=6.39 $, $\beta_3=6.08 $, $\beta_4=5.82 $ (or $X \sim Lognormal(0,2)$)
\end{enumerate}

\noindent From Tables 5, 6 and 7, we can see that empirical power of the JEL based test higher than other tests.

\doublespace
\begin{singlespace}
\begin{landscape}
\begin{table}[]
\centering
\caption{Monte-Carlo Variance times the sample size for each estimator: $W_1=n Var(\overline{\beta}_r)$, $W_2=n Var(\widehat{\beta}_r)$, $W_3=n Var(\widehat{\beta}_{r,jack})$ and $W_4=n Var(\widehat{\beta}_{r,adjjack})$.}
\label{}
\begin{tabular}{c|c|cccccccccccccc} \hline
&     & \multicolumn{4}{c}{exp(1)}            &  & \multicolumn{4}{c}{Normal(0,1)}       &  & \multicolumn{4}{c}{log-normal(0,1)}   \\ \cline{3-6} \cline{8-11} \cline{13-16}
 & $n$   & $W_1$  & $W_2$    & $W_3$ & $W_4$ &  & $W_1$  & $W_2$    & $W_3$ & $W_4$ &  & $W_1$  & $W_2$    & $W_3$ & $W_4$ \\ \cline{2-6} \cline{8-11} \cline{13-16}
\multirow{4}{*}{$r$=1} & 25  & 0.5965 & 0.5665 & 0.5831    & 0.4719  &  & 0.3136 & 0.2929 & 0.2963    & 0.2398  &  & 3.6880 & 3.7375 & 3.6053    & 3.5759  \\
                     & 50  & 0.5885 & 0.5736 & 0.5819    & 0.5164  &  & 0.2935 & 0.2835 & 0.2850    & 0.2529  &  & 3.6245 & 3.6538 & 3.5490    & 3.5917  \\
                     & 150 & 0.5709 & 0.5661 & 0.5688    & 0.5427  &  & 0.2980 & 0.2946 & 0.2952    & 0.2816  &  & 3.6509 & 3.7183 & 3.6696    & 3.6020  \\
                     & 300 & 0.5994 & 0.5969 & 0.5983    & 0.5831  &  & 0.2933 & 0.2916 & 0.2919    & 0.2844  &  & 3.5998  & 3.5548  & 3.4058    & 3.5104  \\
&     &        &        &           &         &  &        &        &           &         &  &        &        &           &         \\
\multirow{4}{*}{$r$=2} & 25  & 0.4266 & 0.3884 & 0.4091    & 0.3311  &  & 0.1690 & 0.1519 & 0.1556    & 0.1259  &  & 2.8852 & 2.6441 & 2.8250    & 2.2864  \\
                     & 50  & 0.4119 & 0.3929 & 0.4031    & 0.3577  &  & 0.1618 & 0.1534 & 0.1553    & 0.1378  &  & 3.0319 & 2.9030 & 3.0012    & 2.6634  \\
                     & 150 & 0.4255 & 0.4188 & 0.4224    & 0.4030  &  & 0.1583 & 0.1555 & 0.1561    & 0.1490  &  & 3.0024 & 2.9592 & 2.9922    & 2.8549  \\
                     & 300 & 0.4155 & 0.4123 & 0.4140    & 0.4035  &  & 0.1582 & 0.1568 & 0.1571    & 0.1531  &  & 3.0001 & 2.9785 & 2.9949    & 2.9188  \\
&     &        &        &           &         &  &        &        &           &         &  &        &        &           &         \\
\multirow{4}{*}{$r$=3} & 25  & 0.3433 & 0.3000 & 0.3227    & 0.2612  &  & 0.1128 & 0.0976 & 0.1014    & 0.0821  &  & 2.8900 & 2.5455 & 2.8099    & 2.2742  \\
                     & 50  & 0.3306 & 0.3090 & 0.3203    & 0.2842  &  & 0.1068 & 0.0993 & 0.1012    & 0.0898  &  & 2.6063 & 2.4449 & 2.5654    & 2.2766  \\
                     & 150 & 0.3130 & 0.3059 & 0.3096    & 0.2954  &  & 0.1008 & 0.0983 & 0.0989    & 0.0944  &  & 2.6165 & 2.5615 & 2.6034    & 2.4840  \\
                     & 300 & 0.3263 & 0.3226 & 0.3246    & 0.3163  &  & 0.1054 & 0.1042 & 0.1045    & 0.1018  &  & 2.7465 & 2.7175 & 2.7397    & 2.6701  \\
&     &        &        &           &         &  &        &        &           &         &  &        &        &           &         \\
\multirow{4}{*}{$r$=4} & 25  & 0.2880 & 0.2419 & 0.2659    & 0.2152  &  & 0.0865 & 0.0720 & 0.0759    & 0.0614  &  & 2.7197 & 2.3014 & 2.6157    & 2.1170  \\
                     & 50  & 0.2700 & 0.2473 & 0.2590    & 0.2298  &  & 0.0800 & 0.0729 & 0.0747    & 0.0663  &  & 2.6946 & 2.4782 & 2.6426    & 2.3452  \\
                     & 150 & 0.2560 & 0.2486 & 0.2524    & 0.2408  &  & 0.0741 & 0.0718 & 0.0725    & 0.0691  &  & 2.3813 & 2.3156 & 2.3653    & 2.2567  \\
                     & 300 & 0.2672 & 0.2633 & 0.2653    & 0.2586  &  & 0.0722 & 0.0711 & 0.0714    & 0.0696  &  & 2.3917 & 2.3584 & 2.3834    & 2.3228 \\ \hline
\end{tabular}
\end{table}
\end{landscape}
\end{singlespace}
  \begin{singlespace}
\begin{landscape}
\begin{table}[]
\centering
\caption{{Coverage probability of likelihood ratio tests  for different $r$ and different distributions}}
\label{my-label}
\begin{tabular}{cc|cccccccccccccc} \hline
                     & $n$    & \multicolumn{4}{c}{Exp(1)}     &  & \multicolumn{4}{c}{Normal(0,1)}  &  & \multicolumn{4}{c}{Lognormal(0,1)} \\ \cline{2-6}  \cline{8-11} \cline{13-16}
                     &     & DNEL    & VLX   & JEL   & AJEL &  & DNEL    & VLX   & JEL   & AJEL &  & DNEL     & VLX    & JEL   & AJEL \\ \cline{3-6}  \cline{8-11} \cline{13-16}
\multirow{5}{*}{$r$=1} &25  & 0.913  &0.908	 &0.902	 &0.905	&	&0.919	&0.928	&0.912	&0.953	&& 0.869 & 0.853 & 0.851 & 0.839 \\
                       &50  & 0.939	 &0.933	 &0.932	 &0.934	&   &0.943	&0.948	&0.940	&0.960	&& 0.903 & 0.891 & 0.889 & 0.890  \\
					   &100 & 0.941	 &0.938	 &0.936	 &0.947	&   &0.950	&0.951	&0.947	&0.958	&& 0.925 & 0.916 & 0.918 & 0.917 \\
					   &200 & 0.943	 &0.942	 &0.942	 &0.946	&   &0.949	&0.951	&0.950	&0.957	&& 0.933 & 0.925 & 0.925 & 0.927 \\
					   &300 & 0.946	 &0.946	 &0.946	 &0.950	&   &0.952	&0.950	&0.948	&0.953	&& 0.935 & 0.933 & 0.932 & 0.930  \\
                     &     &       &       &       &        &  &       &       &       &        &  &        &        &       &        \\
\multirow{5}{*}{$r$=2} & 25  & 0.903	&0.951	&0.881	&0.891	&&0.912	&0.967	&0.906	&0.947	&&	0.862	&0.948	&0.839	& 0.826  \\
					   &50   & 0.935	&0.924	&0.924	&0.932	&&	0.94	&0.948	&0.932	&0.961	&&	0.898	&0.886	&0.881	&0.876\\
					   &100  & 0.945	&0.939	&0.938	&0.948  &&	0.945	&0.949	&0.941	&0.958	&&	0.929	&0.912	&0.912	&0.912\\
					   &200  & 0.949	&0.948	&0.945	&0.954  &&	0.950	&0.951	&0.949	&0.956	&&	0.935	&0.924	&0.923	&0.922\\
					   &300  & 0.951	&0.948	&0.949	&0.954  &&	0.944	&0.941	&0.943	&0.947	&&	0.933	&0.924	&0.923	&0.927 \\
                     &     &       &       &       &        &  &       &       &       &        &  &        &        &       &        \\
\multirow{5}{*}{$r$=3} &25   & 0.903	    &0.888	&0.878	&0.892	&&0.903	    &0.927	&0.897	&0.941	&&	0.855	&0.831	&0.829	&0.819 \\
					   &50   & 0.930	    &0.874	&0.911	&0.926	&&0.936		&0.916	&0.926	&0.957	&&	0.899	&0.878	&0.873	&0.876\\
					   &100  & 0.934	&0.915	&0.925	&0.935	&&0.946		&0.900	&0.937	&0.957	&&	0.917	&0.885	&0.899	&0.902\\
					   &200  & 0.950	    &0.951	&0.944	&0.950	&&0.947		&0.908	&0.944	&0.955	&&	0.936	&0.915	&0.922	&0.920\\
					   &300  & 0.947	&0.942	&0.940	&0.947	&&0.949		&0.913	&0.948	&0.957	&&	0.931	&0.915	&0.917	&0.921 \\
                     &     &       &       &       &        &  &       &       &       &        &  &        &        &       &        \\
\multirow{5}{*}{$r$=4} &25  & 0.837	&0.881	&0.858	&0.884	&&	0.897	&0.929	&0.887	&0.932	&&	0.846	&0.866	&0.804 &	0.808\\
					   &50  & 0.926	&0.843	&0.903	&0.922 	&&	0.918	&0.808	&0.909	&0.945  &&  0.888	&0.821	&0.86	& 0.864\\
					   &100 & 0.935	&0.925	&0.927	&0.938 	&&	0.937	&0.894	&0.931	&0.952	&&	0.914	&0.878	&0.892	& 0.890\\
					   &200 & 0.947	&0.944	&0.940	&0.950  &&	0.947	&0.952	&0.942	&0.954	&&	0.926	&0.922	&0.914	& 0.913 \\
					   &300 & 0.942	&0.944	&0.934	&0.941	&&	0.947	&0.948	&0.945	&0.956	&&	0.937	&0.925	&0.927	& 0.926  \\ \hline
\end{tabular}
\end{table}

\begin{table}[]
\centering
\caption{Average length}
\label{my-label}
\begin{tabular}{cc|cccccccccccccc} \hline
                     & $n$   & \multicolumn{4}{c}{Exp(1)}        &  & \multicolumn{4}{c}{Normal(0,1)}     &      & \multicolumn{4}{c}{Lognormal(0,1)}  \\ \cline{2-6}  \cline{8-11} \cline{13-16}
                     &     & DNEL     & VLX    & JEL    & AJEL &  & DNEL     & VLX    & JEL    & AJEL   &      & DNEL     & VLX    & JEL    & AJEL \\ \cline{3-6}  \cline{8-11} \cline{13-16}
\multirow{5}{*}{$r$=1} & 25  & 0.9776 & 0.8696 & 0.8388 & 0.8743 &  & 0.6463 & 0.6178 & 0.5473 & 0.6403 &      & 3.5341 & 3.5113 & 3.4718 & 3.7529 \\
                     & 50  & 0.6737 & 0.5592 & 0.5475 & 0.5816 &  & 0.3700 & 0.3618 & 0.3556 & 0.3557   &      & 2.6065 & 2.5382 & 2.5306 & 2.3840 \\
                     & 100 & 0.4448 & 0.3504 & 0.3478 & 0.3837 &  & 0.2523 & 0.2420 & 0.2358 & 0.2413   &      & 1.3274 & 1.2837 & 1.2805 & 1.2168 \\
                     & 200 & 0.3060 & 0.2643 & 0.2549 & 0.2530 &  & 0.1751 & 0.1748 & 0.1711 & 0.1626   &      & 0.9948 & 0.8870 & 0.8821 & 0.9110 \\
                     & 300 & 0.2446 & 0.1957 & 0.1933 & 0.2063 &  & 0.1422 & 0.1322 & 0.1306 & 0.1313   &      & 0.6882 & 0.6394 & 0.6303 & 0.7153 \\
                     &     &        &        &        &        &  &        &        &        &          &      &        &        &        &        \\
\multirow{5}{*}{$r$=2} & 25  & 0.9031 & 0.6514 & 0.6409 & 0.6859 &  & 0.4788 & 0.4745 & 0.4173 & 0.4419 &      & 3.6222 & 3.5692 & 3.5686 & 3.1831 \\
                     & 50  & 0.6508 & 0.4793 & 0.4733 & 0.4908 &  & 0.3280 & 0.2651 & 0.2524 & 0.2871   &      & 2.0383 & 1.7987 & 1.4620 & 1.9150 \\
                     & 100 & 0.4338 & 0.3209 & 0.3053 & 0.3390 &  & 0.2247 & 0.1854 & 0.1782 & 0.1946   &      & 1.3018 & 1.1537 & 1.1531 & 1.1154 \\
                     & 200 & 0.3141 & 0.2205 & 0.2193 & 0.2176 &  & 0.1616 & 0.1221 & 0.1203 & 0.1214   &      & 0.9575 & 0.8599 & 0.8539 & 0.8384 \\
                     & 300 & 0.2383 & 0.1822 & 0.1806 & 0.1813 &  & 0.1254 & 0.0978 & 0.0956 & 0.0994   &      & 0.7941 & 0.7565 & 0.7508 & 0.7294 \\
                     &     &        &        &        &        &  &        &        &        &        &        &        &        &        &        \\
\multirow{5}{*}{$r$=3} & 25  & 0.9111 & 0.7413 & 0.7188 & 0.6450 &  & 0.4431 & 0.3817 & 0.3778 & 0.3769 &      & 2.8677 & 3.2130 & 2.6444 & 2.9005 \\
                     & 50  & 0.6432 & 0.5342 & 0.4397 & 0.4273 &  & 0.3292 & 0.2688 & 0.2416 & 0.2579   &      & 2.0338 & 1.9462 & 1.9358 & 2.3499 \\
                     & 100 & 0.4622 & 0.2875 & 0.2748 & 0.3007 &  & 0.2256 & 0.2220 & 0.1519 & 0.1530   &      & 1.4327 & 1.3802 & 1.3180 & 1.2739 \\
                     & 200 & 0.3025 & 0.2023 & 0.1943 & 0.1956 &  & 0.1519 & 0.1178 & 0.0999 & 0.1024   &      & 0.9472 & 0.8419 & 0.8336 & 0.8184 \\
                     & 300 & 0.2457 & 0.1625 & 0.1499 & 0.1606 &  & 0.1207 & 0.0900 & 0.0804 & 0.0816   &      & 0.6521 & 0.5749 & 0.5667 & 0.5585 \\
                     &     &        &        &        &        &  &        &        &        &          &      &        &        &        &        \\
\multirow{5}{*}{$r$=4} & 25  & 0.8023 & 0.5576 & 0.5485 & 0.6734 &  & 0.4546 & 0.2925 & 0.2922 & 0.3064 &      & 3.3343 & 2.6987 & 2.3671 & 2.7833 \\
                     & 50  & 0.5911 & 0.3720 & 0.3659 & 0.4000 &  & 0.3216 & 0.1965 & 0.1918 & 0.2022   &      & 1.6487 & 1.5862 & 1.5822 & 2.0474 \\
                     & 100 & 0.4308 & 0.2940 & 0.2923 & 0.3067 &  & 0.2104 & 0.1360 & 0.1346 & 0.1303   &      & 1.5489 & 1.2440 & 1.2403 & 1.3687 \\
                     & 200 & 0.2997 & 0.1719 & 0.1691 & 0.1796 &  & 0.1426 & 0.1174 & 0.0900 & 0.0922   &      & 0.9711 & 0.9209 & 0.8783 & 0.8510 \\
                     & 300 & 0.2333 & 0.1352 & 0.1302 & 0.1422 &  & 0.1145 & 0.0806 & 0.0760 & 0.0797   &      & 0.7077 & 0.5738 & 0.5716 & 0.5934 \\ \hline
\end{tabular}
\end{table}
\end{landscape}
\end{singlespace}

\begin{singlespace}
\begin{landscape}
\begin{table}[]
\centering
\caption{Size of the likelihood ratio tests for different $r$ }
\label{my-label}
\begin{tabular}{cc|cccccccccccccc} \hline
                     & $n$    & \multicolumn{4}{c}{Exp(1)}     &  & \multicolumn{4}{c}{Normal(0,1)}  &  & \multicolumn{4}{c}{Lognormal(0,1)} \\ \cline{2-6}  \cline{8-11} \cline{13-16}
                     &     & DNEL    & VLX   & JEL   & AJEL &  & DNEL    & VLX   & JEL   & AJEL&  & DNEL     & VLX    & JEL   & AJEL \\ \cline{3-6}  \cline{8-11} \cline{13-16}
\multirow{5}{*}{$r$=1} & 25  & 0.095 & 0.092 & 0.098 & 0.087  &  & 0.081 & 0.072 & 0.088 & 0.047  &  & 0.131  & 0.147  & 0.149 & 0.161  \\
                     & 50  & 0.061 & 0.067 & 0.068 & 0.066  &  & 0.057 & 0.052 & 0.060 & 0.040  &  & 0.097  & 0.109  & 0.111 & 0.110  \\
                     & 100 & 0.059 & 0.062 & 0.064 & 0.053  &  & 0.050 & 0.049 & 0.053 & 0.042  &  & 0.075  & 0.084  & 0.082 & 0.083  \\
                     & 200 & 0.057 & 0.058 & 0.058 & 0.054  &  & 0.051 & 0.049 & 0.050 & 0.043  &  & 0.067  & 0.075  & 0.075 & 0.073  \\
                     & 300 & 0.054 & 0.054 & 0.054 & 0.050  &  & 0.048 & 0.047 & 0.052 & 0.050  &  & 0.065  & 0.067  & 0.068 & 0.070  \\
                     &     &       &       &       &        &  &       &       &       &        &  &        &        &       &        \\
\multirow{5}{*}{$r$=2} & 25  & 0.097 & 0.049 & 0.119 & 0.109  &  & 0.088 & 0.063 & 0.094 & 0.053  &  & 0.138  & 0.052  & 0.161 & 0.174  \\
                     & 50  & 0.065 & 0.076 & 0.076 & 0.065  &  & 0.060 & 0.052 & 0.068 & 0.039  &  & 0.102  & 0.114  & 0.119 & 0.124  \\
                     & 100 & 0.055 & 0.061 & 0.062 & 0.052  &  & 0.055 & 0.051 & 0.059 & 0.042  &  & 0.071  & 0.088  & 0.088 & 0.088  \\
                     & 200 & 0.051 & 0.052 & 0.055 & 0.046  &  & 0.050 & 0.049 & 0.051 & 0.044  &  & 0.065  & 0.076  & 0.077 & 0.078  \\
                     & 300 & 0.049 & 0.052 & 0.051 & 0.046  &  & 0.056 & 0.059 & 0.057 & 0.053  &  & 0.067  & 0.076  & 0.077 & 0.073  \\
                     &     &       &       &       &        &  &       &       &       &        &  &        &        &       &        \\
\multirow{5}{*}{$r$=3} & 25  & 0.100 & 0.112 & 0.122 & 0.108  &  & 0.097 & 0.073 & 0.103 & 0.059  &  & 0.145  & 0.169  & 0.171 & 0.181  \\
                     & 50  & 0.070 & 0.126 & 0.089 & 0.074  &  & 0.064 & 0.084 & 0.074 & 0.043  &  & 0.101  & 0.122  & 0.127 & 0.124  \\
                     & 100 & 0.066 & 0.085 & 0.075 & 0.065  &  & 0.054 & 0.100 & 0.063 & 0.043  &  & 0.083  & 0.115  & 0.101 & 0.098  \\
                     & 200 & 0.054 & 0.058 & 0.060 & 0.053  &  & 0.053 & 0.092 & 0.056 & 0.045  &  & 0.064  & 0.085  & 0.078 & 0.080  \\
                     & 300 & 0.051 & 0.049 & 0.056 & 0.050  &  & 0.051 & 0.087 & 0.052 & 0.043  &  & 0.069  & 0.085  & 0.083 & 0.079  \\
                     &     &       &       &       &        &  &       &       &       &        &  &        &        &       &        \\
\multirow{5}{*}{$r$=4} & 25  & 0.116 & 0.119 & 0.142 & 0.121  &  & 0.103 & 0.071 & 0.113 & 0.068  &  & 0.154  & 0.134  & 0.196 & 0.192  \\
                     & 50  & 0.074 & 0.157 & 0.097 & 0.078  &  & 0.082 & 0.192 & 0.091 & 0.055  &  & 0.112  & 0.179  & 0.140 & 0.136  \\
                     & 100 & 0.062 & 0.075 & 0.073 & 0.065  &  & 0.063 & 0.106 & 0.069 & 0.048  &  & 0.086  & 0.122  & 0.108 & 0.110  \\
                     & 200 & 0.058 & 0.056 & 0.060 & 0.059  &  & 0.053 & 0.048 & 0.058 & 0.046  &  & 0.074  & 0.078  & 0.086 & 0.087  \\
                     & 300 & 0.050 & 0.056 & 0.056 & 0.053  &  & 0.053 & 0.052 & 0.055 & 0.044  &  & 0.063  & 0.075  & 0.073 & 0.074  \\ \hline
\end{tabular}
\end{table}

\begin{table}[]
\centering
\caption{Power of the likelihood ratio tests for exponential distribution}
\label{my-label}
\begin{tabular}{ccccccccccc} \hline
                     &     & \multicolumn{4}{c}{$H_1:$ $X\sim$ Exp(0.9)} &  & \multicolumn{4}{c}{$H_1:$ $X\sim$Exp(0.8)} \\ \cline{2-6} \cline{8-11}
                     & $n$   & DNEL       & VLX     & JEL     & AJEL   &  & DNEL       & VLX     & JEL     & AJEL   \\ \cline{3-6} \cline{8-11}
\multirow{5}{*}{$r$=1} & 25  & 0.088    & 0.099   & 0.126   & 0.041    &  & 0.147    & 0.215   & 0.280   & 0.018    \\
                     & 50  & 0.097    & 0.124   & 0.152   & 0.028    &  & 0.233    & 0.359   & 0.409   & 0.092    \\	
                     & 100 & 0.124    & 0.184   & 0.210   & 0.069    &  & 0.386    & 0.608   & 0.642   & 0.386    \\
                     & 200 & 0.202    & 0.328   & 0.353   & 0.204    &  & 0.643    & 0.868   & 0.881   & 0.784    \\
                     & 300 & 0.275    & 0.440   & 0.462   & 0.337    &  & 0.808    & 0.965   & 0.968   & 0.940    \\
                     &     &          &         &         &          &  &          &         &         &          \\
\multirow{5}{*}{$r$=2} & 25  & 0.092    & 0.072   & 0.152   & 0.047    &  & 0.131    & 0.161   & 0.301   & 0.023    \\
                     & 50  & 0.091    & 0.127   & 0.180   & 0.033    &  & 0.178    & 0.340   & 0.449   & 0.100    \\
                     & 100 & 0.107    & 0.190   & 0.235   & 0.080    &  & 0.304    & 0.600   & 0.670   & 0.407    \\
                     & 200 & 0.159    & 0.307   & 0.347   & 0.206    &  & 0.512    & 0.866   & 0.885   & 0.791    \\
                     & 300 & 0.215    & 0.443   & 0.476   & 0.343    &  & 0.674    & 0.965   & 0.971   & 0.942    \\
                     &     &          &         &         &          &  &          &         &         &          \\
\multirow{5}{*}{$r$=3} & 25  & 0.103    & 0.152   & 0.186   & 0.057    &  & 0.123    & 0.267   & 0.328   & 0.027    \\
                     & 50  & 0.082    & 0.171   & 0.176   & 0.035    &  & 0.153    & 0.409   & 0.442   & 0.100    \\
                     & 100 & 0.095    & 0.204   & 0.226   & 0.081    &  & 0.235    & 0.611   & 0.645   & 0.383    \\
                     & 200 & 0.128    & 0.300   & 0.333   & 0.201    &  & 0.407    & 0.864   & 0.885   & 0.780    \\
                     & 300 & 0.175    & 0.429   & 0.461   & 0.338    &  & 0.553    & 0.958   & 0.965   & 0.936    \\
\multicolumn{1}{l}{} &     &          &         &         &          &  &          &         &         &          \\
\multirow{5}{*}{r=4} & 25  & 0.102    & 0.178   & 0.188   & 0.061    &  & 0.121    & 0.351   & 0.353   & 0.030    \\
                     & 50  & 0.088    & 0.233   & 0.295   & 0.042    &  & 0.131    & 0.453   & 0.455   & 0.106    \\
                     & 100 & 0.087    & 0.238   & 0.322   & 0.076    &  & 0.202    & 0.626   & 0.652   & 0.389    \\
                     & 200 & 0.123    & 0.354   & 0.386   & 0.215    &  & 0.344    & 0.856   & 0.879   & 0.776    \\
                     & 300 & 0.156    & 0.440   & 0.462   & 0.337    &  & 0.472    & 0.949   & 0.959   & 0.924    \\ \hline
\end{tabular}
\end{table}

\begin{table}[]
\centering
\caption{Power of the likelihood ratio tests for Normal distribution}
\label{}
\begin{tabular}{ccccccccccc} \hline
                     &     & \multicolumn{4}{c}{$H_1: X\sim$ Normal(0,$2^2$)} &  & \multicolumn{4}{c}{$H_1: X\sim$ Normal(0,$3^2$)} \\ \cline{2-6} \cline{8-11}
                     & $n$   & DNEL      & VLX     & JEL     & AJEL   &  & DNEL      & VLX     & JEL   & AJEL  \\ \cline{3-6} \cline{8-11}
\multirow{5}{*}{$r$=1} & 25 & 	0.291 &	0.261&	0.380&	0.371&&		0.509&	0.434&	0.554&	0.421 \\
					 & 50 & 	0.495 &	0.462&	0.561&	0.545&&		0.781&	0.717&	0.783&	0.725 \\
					 & 100&	0.782 &	0.764&	0.829&	0.808&&		0.972&	0.952&	0.965&	0.955  \\
					 & 200&	0.970 &	0.965&	0.984&	0.973&&		1.000&	0.999&	1.000&	1.000  \\
					 & 300&	0.997 &	0.997&	0.998&	0.997&&		1.000&	1.000&	1.000&	1.000  \\
					 &     &         &         &         &          &  &         &         &        &         \\
\multirow{5}{*}{$r$=2} &25  & 0.403&	0.523&	0.648&	0.645&&		0.680&	0.692&	0.838&	0.690 \\
					 &50  & 0.649&	0.700&	0.854&	0.847&&		0.921&	0.944&	0.974&	0.955\\
					 &100 & 0.898&	0.948&	0.989&	0.978&&		0.998&	1.000&	1.000&	1.000\\
					 &200 & 0.994&	1.000&	1.000&	1.000&&		1.000&	1.000&	1.000&	1.000\\	
					 &300 & 1.000&	1.000&	1.000&	1.000&&		1.000&	1.000&	1.000&	1.000\\
                     &     &         &         &         &          &  &         &         &        &         \\

\multirow{5}{*}{$r$=3} &25& 	0.421&	0.479&	0.686&	0.683&&		0.684&	0.810&	0.932&	0.805\\
					&50&	0.641&	0.861&	0.931&	0.927&&		0.914&	0.985&	0.995&	0.989\\
					&100&	0.889&	0.991&	0.998&	0.996&&		0.996&	1.000&	1.000&	1.000\\
					&200&	0.993&	1.000&	1.000&	1.000&&		1.000&	1.000&	1.000&	1.000\\
					&300&	1.000&	1.000&	1.000&	1.000&&		1.000&	1.000&	1.000&	1.000\\
                     &     &         &         &         &          &  &         &         &        &         \\
\multirow{5}{*}{$r$=4} &25&	0.392&	0.636&	0.794&	0.756&&		0.653&	0.911&	0.960&	0.860\\
					&50&	0.602&	0.941&	0.997&	0.965&&		0.892&	0.996&	0.999&	0.996\\
					&100&	0.862&	0.998&	1.000&	0.999&&		0.992&	1.000&	1.000&	1.000\\
					&200&	0.990&	1.000&	1.000&	1.000&&		0.993&	1.000&	1.000&	1.000\\
					&300&	0.996&	1.000&	1.000&	1.000&&		1.000&	1.000&	1.000&	1.000\\ \hline
\end{tabular}
\end{table}

\begin{table}[]
\centering
\caption{Power of the likelihood ratio tests for log normal distribution}
\label{}
\begin{tabular}{ccccccccccc} \hline
                     &     & \multicolumn{4}{c}{$H_1: X\sim$ Lognormal(0,1.5)} &  & \multicolumn{4}{c}{$H_1: X\sim$ Lognormal(0,2)} \\ \cline{2-6} \cline{8-11}
                     & $n$   & DNEL      & VLX     & JEL     & AJEL   &  & DNEL      & VLX     & JEL    & AJEL  \\ \cline{3-6} \cline{8-11}
\multirow{5}{*}{$r$=1} & 25  & 0.378   & 0.458   & 0.496   & 0.159    &  & 0.742   & 0.822   & 0.842  & 0.553   \\
                     & 50  & 0.592   & 0.704   & 0.730   & 0.525    &  & 0.940   & 0.972   & 0.974  & 0.942   \\
                     & 100 & 0.856   & 0.928   & 0.935   & 0.885    &  & 0.997   & 0.999   & 0.999  & 0.999   \\
                     & 200 & 0.984   & 0.996   & 0.997   & 0.994    &  & 1.000   & 1.000   & 1.000  & 1.000   \\
                     & 300 & 0.999   & 1.000   & 1.000   & 1.000    &  & 1.000   & 1.000   & 1.000  & 1.000   \\
                     &     &         &         &         &          &  &         &         &        &         \\
\multirow{5}{*}{$r$=2} & 25  & 0.359   & 0.409   & 0.553   & 0.198    &  & 0.701   & 0.679   & 0.866  & 0.592   \\
                     & 50  & 0.552   & 0.702   & 0.766   & 0.580    &  & 0.916   & 0.872   & 0.979  & 0.957   \\
                     & 100 & 0.828   & 0.940   & 0.949   & 0.913    &  & 0.997   & 0.940   & 1.000  & 0.999   \\
                     & 200 & 0.979   & 0.997   & 0.998   & 0.996    &  & 1.000   & 0.972   & 1.000  & 1.000   \\
                     & 300 & 0.998   & 1.000   & 1.000   & 1.000    &  & 1.000   & 0.984   & 1.000  & 1.000   \\
                     &     &         &         &         &          &  &         &         &        &         \\
\multirow{5}{*}{$r$=3} & 25  & 0.341   & 0.432   & 0.580   & 0.228    &  & 0.696   & 0.507   & 0.894  & 0.649   \\
                     & 50  & 0.532   & 0.721   & 0.798   & 0.622    &  & 0.912   & 0.630   & 0.986  & 0.968   \\
                     & 100 & 0.794   & 0.932   & 0.955   & 0.924    &  & 0.994   & 0.715   & 1.000  & 1.000   \\
                     & 200 & 0.969   & 0.994   & 0.999   & 0.998    &  & 1.000   & 0.806   & 1.000  & 1.000   \\
                     & 300 & 0.996   & 0.998   & 1.000   & 1.000    &  & 1.000   & 0.853   & 1.000  & 1.000   \\
                     &     &         &         &         &          &  &         &         &        &         \\
\multirow{5}{*}{$r$=4} & 25  & 0.335   & 0.530   & 0.606   & 0.246    &  & 0.663   & 0.503   & 0.908  & 0.661   \\
                     & 50  & 0.505   & 0.780   & 0.806   & 0.644    &  & 0.884   & 0.639   & 0.990  & 0.972   \\
                     & 100 & 0.759   & 0.943   & 0.962   & 0.931    &  & 0.990   & 0.739   & 1.000  & 1.000   \\
                     & 200 & 0.956   & 0.995   & 1.000   & 0.999    &  & 1.000   & 0.823   & 1.000  & 1.000   \\
                     & 300 & 0.973   & 0.998   & 1.000   & 1.000    &  & 1.000   & 0.962   & 1.000  & 1.000 \\ \hline
\end{tabular}
\end{table}
\end{landscape}
\end{singlespace}

\section{Application}
\par Modelling rain fall frequency values is an important area of research for  planning various  projects such as construction of bridges and spillways of dams. Many probability distributions are used for this modelling purpose and this includes using method of PWM for the estimation of the parameters of interest. Here we look into  95 \% interval estimates of $\beta_r$ for rain fall data set of Indian States.
\par Government of India  releases data set related to various sectors and the same can be accessible through the web page \url{www.data.gov.in}. These sectors contain agriculture, education, health and family welfare, labour and employment, travel and tourism etc. We use weighted average of monthly rain fall (in mm) data for the whole country starting from 1901 to 2014 released by Meteorological Department, Ministry of Earth Sciences, India. This data set is based on more than 2000 rain guage readings spread over the entire country. During the span of 1901 to 2014, rain fall reached maximum for the year 1917 and  recorded  minimum frequency for the year 1972. India had witnessed severe drought in the year 1972-73 in drought prone areas especially large parts of Maharashtra, a state of India.
\par Table 8 gives 95 \% interval estimates of $\beta_1$ and corresponding average lengths for the rain fall data. From the data, it is clear that for all most all months JEL interval has shorter average length.  Also performances of JEL interval is comparable with VLX interval except for the month September where VLX interval is less stable. For AJEL based approach, the intervals are more wide for the months June to September. 
\begin{singlespace}
\begin{landscape}
\begin{table}[]
\centering
\caption{ 95\% confidence interval estimators of $\beta_1$ and average lengths for rainfall data}
\label{my-label}
\begin{tabular}{|l|l|l|l|l|}
\hline
Month     & DNEL                                                                                       & VLX                                                                            & JEL                                                                             & AJEL                                                                               \\ \hline
January   & \multicolumn{1}{c|}{\begin{tabular}[c]{@{}c@{}}(10.2040, 14.3500)\\ 4.1462\end{tabular}} & \begin{tabular}[c]{@{}c@{}}(11.0018,13.3237) \\           2.3219\end{tabular}  & \begin{tabular}[c]{@{}c@{}}(11.0817, 13.3838) \\            {2.3021} \end{tabular} & \begin{tabular}[c]{@{}c@{}}(10.2383, 13.0992) \\            2.8609\end{tabular}      \\ \hline
February  & \begin{tabular}[c]{@{}c@{}}(12.4763, 17.2933)\\          4.8170\end{tabular}             & \begin{tabular}[c]{@{}c@{}}(13.4675, 15.9107) \\          2.4433\end{tabular}  & \begin{tabular}[c]{@{}c@{}}(13.5546, 15.9952) \\         {2.4406} \end{tabular}    & \begin{tabular}[c]{@{}c@{}}(12.5195, 15.6855)\\               3.1661\end{tabular}    \\ \hline
March     & \begin{tabular}[c]{@{}c@{}}(14.5653, 20.1041)\\          5.5387\end{tabular}             & \begin{tabular}[c]{@{}c@{}}(15.7437, 18.5202)\\         2.7765\end{tabular}    & \begin{tabular}[c]{@{}c@{}}(15.8361, 18.6097) \\          {2.7736} \end{tabular}   & \begin{tabular}[c]{@{}c@{}}(14.5958, 18.2250)\\                3.6291\end{tabular}   \\ \hline
April     & \begin{tabular}[c]{@{}c@{}}(18.9144, 24.9451)\\         6.0308\end{tabular}              & \begin{tabular}[c]{@{}c@{}}(20.5914, 22.8184) \\        2.2270\end{tabular}    & \begin{tabular}[c]{@{}c@{}}(20.6697, 22.8892) \\            {2.2195} \end{tabular} & \begin{tabular}[c]{@{}c@{}}(18.7482, 22.4825)\\                3.7343\end{tabular}   \\ \hline
May       & \begin{tabular}[c]{@{}c@{}}(31.1654, 41.1201) \\         9.9546\end{tabular}             & \begin{tabular}[c]{@{}c@{}}(33.9171, 37.6333) \\        {3.7192} \end{tabular}    & \begin{tabular}[c]{@{}c@{}}(34.0352, 37.7536) \\           3.7184\end{tabular}  & \begin{tabular}[c]{@{}c@{}}(30.8670, 37.1203) \\              6.2533\end{tabular}    \\ \hline
June      & \begin{tabular}[c]{@{}c@{}}(82.7641, 107.3357)\\         24.5717\end{tabular}            & \begin{tabular}[c]{@{}c@{}}(90.5382, 97.7330)\\         7.1948\end{tabular}    & \begin{tabular}[c]{@{}c@{}}(90.8288, 97.9655)\\            {7.1367} \end{tabular}  & \begin{tabular}[c]{@{}c@{}}(81.7213, 96.3148) \\            14.5935\end{tabular}     \\ \hline
July      & \begin{tabular}[c]{@{}c@{}}(138.682, 175.931)\\          37.2482\end{tabular}            & \begin{tabular}[c]{@{}c@{}}(152.4140, 159.2123)\\          6.7983\end{tabular} & \begin{tabular}[c]{@{}c@{}}(152.696, 159.460)\\            {6.7641} \end{tabular}  & \begin{tabular}[c]{@{}c@{}}(135.8374, 156.8447) \\            21.0063\end{tabular}   \\ \hline
August    & \begin{tabular}[c]{@{}c@{}}(122.7836, 156.2862) \\         33.5026\end{tabular}          & \begin{tabular}[c]{@{}c@{}}(134.8720, 141.3698) \\         6.4978\end{tabular} & \begin{tabular}[c]{@{}c@{}}(135.1352, 141.6145)\\           {6.4793}\end{tabular} & \begin{tabular}[c]{@{}c@{}}(120.4020, 139.2415) \\              18.8395\end{tabular} \\ \hline
September & \begin{tabular}[c]{@{}c@{}}(84.7445, 109.9898)\\         25.2453\end{tabular}            & \begin{tabular}[c]{@{}c@{}}(92.6231, 108.6973) \\        {16.0743} \end{tabular} & \begin{tabular}[c]{@{}c@{}}(92.9210, 100.4291) \\          {7.5081} \end{tabular}  & \begin{tabular}[c]{@{}c@{}}(83.6334, 98.5444) \\             14.9110\end{tabular}    \\ \hline
October   & \begin{tabular}[c]{@{}c@{}}(39.6969, 53.8395) \\        14.1426\end{tabular}             & \begin{tabular}[c]{@{}c@{}}(42.9838, 49.4825)\\         6.4987\end{tabular}    & \begin{tabular}[c]{@{}c@{}}(43.20343, 49.6888) \\         {6.4854} \end{tabular}   & \begin{tabular}[c]{@{}c@{}}(39.58949, 48.73922)\\              9.1497\end{tabular}   \\ \hline
November  & \begin{tabular}[c]{@{}c@{}}(16.3331, 23.2432)\\          6.9101\end{tabular}             & \begin{tabular}[c]{@{}c@{}}(17.5455, 21.4974)\\           {3.9620} \end{tabular}  & \begin{tabular}[c]{@{}c@{}}(17.6720, 21.6304)\\          3.9584\end{tabular}    & \begin{tabular}[c]{@{}c@{}}(16.4306, 21.1816) \\            4.7510\end{tabular}      \\ \hline
December  & \begin{tabular}[c]{@{}c@{}}(8.2039, 11.9599) \\          3.7561\end{tabular}             & \begin{tabular}[c]{@{}c@{}}(8.776057, 11.1920)\\           2.4160\end{tabular} & \begin{tabular}[c]{@{}c@{}}(8.8470, 11.2497) \\           2.4028 \end{tabular}   & \begin{tabular}[c]{@{}c@{}}(8.21451, 11.0223) \\             2.8078\end{tabular}      \\ \hline
\end{tabular}
\end{table}
\end{landscape}
\end{singlespace}

\section{Conclusions}

 Probability weighted moments  generalize the concept of moments of a probability distribution and are used to estimates parameters involved in model extremes of the natural phenomena. In this article we proposed JEL and AJEL based inference for $\beta_r$. We showed that the log JEL and AJEL ratio statistics is asymptotically distributed as chi-square distribution with one degree of freedom and constructed JEL and AJEL based confidence interval for $\beta_r$.  We also developed JEL and AJEL based tests for testing the hypothesis $\beta_r=\beta_r^0$, where $\beta_r^0$ is a specific value of $\beta_r$. Simulation results shows that proposed JEL and AJEL test perform better than other recently developed tests in terms of empirical power. We also showed that the JEL based confidence interval has minimum length compared to competitor and has coverage probability of desired confidence level.  The method is illustrated by using a rainfall data of Indian states.

\end{document}